\documentclass{revtex4}
\usepackage{amsmath,amssymb,amsthm,graphicx,bbm}
\usepackage[
	bookmarks = false,
	colorlinks = true,
	linkcolor = blue,
	urlcolor= black,
	citecolor = blue]{hyperref}
\usepackage{braket}
\newtheorem{theorem}{Theorem}
\newtheorem{lemma}{Lemma}
\newtheorem{corollary}{Corollary}
\newtheorem{definition}{Definition}
\newtheorem{proposition}{Proposition}
\newtheorem{example}{Example}
\DeclareMathOperator{\Tr}{Tr}
\DeclareMathOperator{\PT}{PT}
\DeclareMathOperator{\PPT}{PPT}
\DeclareMathOperator{\DEW}{DEW}

\newcommand{\ketbra}[2]{\ket{#1}\!\bra{#2}}
\begin{document}
\title{Bipartite quantum state discrimination and decomposable entanglement witness}
\author{Donghoon Ha}
\affiliation{Department of Applied Mathematics and Institute of Natural Sciences, Kyung Hee University, Yongin 17104, Republic of Korea}
\author{Jeong San Kim}
\email{freddie1@khu.ac.kr}
\affiliation{Department of Applied Mathematics and Institute of Natural Sciences, Kyung Hee University, Yongin 17104, Republic of Korea}

\begin{abstract}
We consider bipartite quantum state discrimination using positive-partial-transpose measurements and show that minimum-error discrimination by positive-partial-transpose measurements
is closely related to entanglement witness.
By using the concept of decomposable entanglement witness,
we establish conditions on minimum-error discrimination by positive-partial-transpose measurements. We also provide conditions on the upper bound of the maximum success probability over all possible positive-partial-transpose measurements.
Finally, we illustrate our results using examples of multidimensional bipartite quantum states.
\end{abstract}
\maketitle
\section{Introduction}
\indent Quantum nonlocality is a fascinating phenomenon that occurs in multipartite quantum systems\cite{horo2009,chid2013,brun2014}. 
In discriminating multipartite quantum states, quantum nonlocality occurs when 
optimal state discrimination, which gives the maximum success probability overall possible measurements, cannot be realized only by \emph{local operations and classical communication}(LOCC)\cite{chit20142,chef2000,berg2007,barn20091,bae2015}. 
Orthogonal quantum states can always be perfectly discriminated using appropriate global measurement. However, there exist some multipartite orthogonal quantum states that cannot be perfectly discriminated only by LOCC measurements\cite{benn19991,ghos2001}.
Moreover, some multipartite nonorthogonal quantum states cannot be optimally discriminated using only LOCC measurements\cite{pere1991,chit2013}.
Nonetheless, characterizing local discrimination of quantum states is still a hard task due to the lack of good mathematical structure for LOCC.\\
\indent The first nonlocal phenomenon in quantum state discrimination was shown through nine $3\otimes3$ pure orthogonal product states\cite{benn19991}.
On the other hand, it was shown that nonlocality does not occur in discriminating any two multipartite pure states\cite{walg2000,virm2001}.
Since then, there have been several studies focused on local indistinguishability or local distinguishability of mutually orthogonal states\cite{lu2010,akib2018,ha20222,cohe2023}.
In particular, the local distinguishability of two $2\otimes2$ pure orthogonal entangled states has been experimentally demonstrated\cite{lu2010}.
It was also shown that the nonlocality of the $N$-fold quantum state ensemble can disappear asymptotically\cite{akib2018}.
Recently, an upper bound of the maximum success probability overall separable measurements was established in the optimization for local unambiguous discrimination of multipartite quantum states\cite{ha20222}.
Moreover, necessary conditions for perfect discrimination by asymptotic LOCC were established in discriminating orthogonal pure states\cite{cohe2023}.\\
\indent The quantum nonlocal phenomenon also arises in the correlations of a multipartite quantum system. 
Entanglement is a well-known quantum correlation that cannot be realized using only LOCC\cite{horo2009}. The nonlocal property of entanglement can be used as a useful resource in various quantum information processing tasks such as quantum cryptography, teleportation, and local discrimination of multipartite quantum states\cite{eker1991,benn1993,chit2019,band2021}. 
For this reason, much attention has been shown for characterizing quantum entanglement\cite{amic2008,guhn2009,kett2020}.\\
\indent An important task of characterizing quantum entanglement is to design ways to detect the existence of entanglement. 
\emph{Entanglement witness}(EW) is an observable having non-negative mean value for any separable state, but negative for some entangled states\cite{horo1996,terh2000,lewe2000,chru2014}.
As EW detects the existence of entanglement that is an important quantum nonlocality, it is natural to ask whether EW can also be used to characterize the limit on local discrimination of quantum states, another important quantum nonlocality.\\
\indent Here, we provide an answer to the question by establishing a specific relation between the properties of EW and \emph{positive-partial-transpose}(PPT) measurements, a mathematically well-structured set of measurements having LOCC measurements as special cases.
By using the concept of decomposability of operators,
we show that the minimum-error discrimination of bipartite quantum states using PPT measurements strongly depends on the existence of \emph{decomposable entanglement witness}(DEW).
More precisely, we establish conditions on minimum-error discrimination by PPT measurements in terms of DEW. We also provide conditions on the maximum success probability over all possible PPT measurements. Finally, we illustrate our results using examples of multidimensional bipartite quantum states.\\
\indent Because PPT measurements have LOCC measurements as special cases, our results provide a useful method to detect the nonlocality arising in quantum state discrimination. 
Moreover, our results can be applied to any ensemble of bipartite states in an arbitrary dimension, whereas the previous results\cite{lu2010,akib2018,cohe2023} are only valid for some restricted cases with certain conditions.
We also note that our results provide a systematic way to construct a bipartite quantum state ensemble showing nonlocality in quantum state discrimination.\\
\indent This paper is organized as follows.
In Section~\ref{sec:pre},
we first recall the definitions and some properties about PPT measurements and DEW. We also recall the definition of minimum-error discrimination as well as some useful properties of the optimal measurements. 
In Section \ref{sec:ppq}, we provide conditions on the maximum success probability over all possible PPT measurements (Theorems~\ref{thm:pptq} and \ref{thm:mnsc}).
In Section~\ref{sec:mep},
we provide conditions on minimum-error discrimination by PPT measurements in terms of DEW(Theorems~\ref{thm:qmsc} and \ref{thm:qupb}). 
Finally, we illustrate our results using examples of multidimensional bipartite quantum states. In Section~\ref{sec:dis}, we summarize our results with possible applications and future works. We also discuss a systematic way to construct a bipartite quantum state ensemble showing nonlocality in quantum state discrimination.

\section{Preliminaries}\label{sec:pre}
\indent For a bipartite Hilbert space $\mathcal{H}=\mathbb{C}^{d_{1}}\otimes\mathbb{C}^{d_{2}}$,
let $\mathbb{H}$ be the set of all Hermitian operators acting on $\mathcal{H}$.
A bipartite quantum state is expressed by a density operator that is $\rho\in\mathbb{H}$ with positive semidefiniteness $\rho\succeq0$ and unit trace $\Tr\rho=1$.
A measurement is represented by a positive operator-valued measure that is $\{M_{i}\}_{i}\subseteq\mathbb{H}$ satisfying 
positive semidefiniteness $M_{i}\succeq0$ for all $i$ and the completeness relation $\sum_{i}M_{i}=\mathbbm{1}$, where $\mathbbm{1}$ is the identity operator in $\mathbb{H}$. 
When a measurement $\{M_{i}\}_{i}$ is performed 
on the quantum state $\rho$, the probability of obtaining the measurement outcome corresponding to $M_{j}$ is $\Tr(\rho M_{j})$.

\begin{definition}
$E\in\mathbb{H}$ is called \emph{PPT} if 
$E^{\PT}\succeq0$,
where the superscript $\PT$ is to indicate
the partial transposition \cite{pere1996,pptp}. 
Similarly, we say that $\{E_{i}\}_{i}\subseteq\mathbb{H}$ is \emph{PPT} if $E_{i}$ is PPT for all $i$. 
\end{definition}
\noindent We denote the set of all positive-semidefinite PPT operators in $\mathbb{H}$ as
\begin{equation}\label{def:pptpd}
\mathbb{PPT}_{+}=\{
E\in\mathbb{H}\,|\,E\succeq 0,~E^{\PT}\succeq 0\},
\end{equation}
and its dual set
as $\mathbb{PPT}_{+}^{*}$, that is,
\begin{equation}\label{eq:pptpsd}
\mathbb{PPT}_{+}^{*}=\{E\in\mathbb{H}\,|\,\Tr(EF)\geqslant0
~\forall F\in\mathbb{PPT}_{+}\}.
\end{equation}
\indent A measurement is called a \emph{LOCC measurement} if it can be implemented by LOCC, and a measurement $\{M_{i}\}_{i}$ is called a \emph{PPT measurement} if $M_{i}$ is PPT for all $i$. We note that every LOCC measurement is a PPT measurement \cite{chit20142}. 

\subsection{Decomposable entanglement witness}
\begin{definition}
$W\in\mathbb{H}$ is called \emph{decomposable} if 
it can be written as a sum of a positive-semidefinite operator and a PPT operator in $\mathbb{H}$, that is,
\begin{equation}\label{eq:wpq}
W=P+Q^{\PT}
\end{equation}
for some $P\succeq0$ and $Q\succeq0$.
\end{definition}

\indent The following proposition provides an
equivalent condition of the decomposability in Eq.~\eqref{eq:wpq} \cite{lewe2000,chru2014}.

\begin{proposition}\label{pro:wptp}
$W\in\mathbb{H}$ is decomposable
if and only if
\begin{equation}
W\in\mathbb{PPT}_{+}^{*}, 
\end{equation}
where $\mathbb{PPT}_{+}^{*}$ is defined in Eq.~\eqref{eq:pptpsd}.
\end{proposition}
\indent A positive-semidefinite operator $E\in\mathbb{H}$ is called \emph{separable} if it can be expressed as a sum of positive-semidefinite product operators, that is,
\begin{equation}
E=\sum_{l}A_{l}\otimes B_{l},
\end{equation}
where $A_{l}$ and $B_{l}$ are positive-semidefinite operators on $\mathbb{C}^{d_{1}}$ and $\mathbb{C}^{d_{2}}$ of $\mathcal{H}$, respectively. 
We denote the set of all positive-semidefinite \emph{separable} operators in $\mathbb{H}$ as
\begin{equation}\label{eq:sepdef}
\mathbb{SEP}=\{E\in\mathbb{H}\,|\, E\succeq0,~E:\mbox{separable}\}.
\end{equation}
We also denote the dual set of $\mathbb{SEP}$ as
\begin{equation}
\mathbb{SEP}^{*}=\{E\in\mathbb{H}\,|\,\Tr(EF)\geqslant0~\forall F\in\mathbb{SEP}\}.
\end{equation}
\begin{definition}\label{def:ew}
$W\in\mathbb{H}$ is called an \emph{EW} if it is in $\mathbb{SEP}^{*}$ but not positive semidefinite, that is,
\begin{equation}
W\in\mathbb{SEP}^{*},~W\not\succeq0.
\end{equation}
In particular, an EW $W\in\mathbb{H}$ is called a \emph{DEW} if
it is decomposable. 
\end{definition}

\begin{figure}[!tt]
\centerline{\includegraphics*[bb=0 0 530 525,scale=0.38]{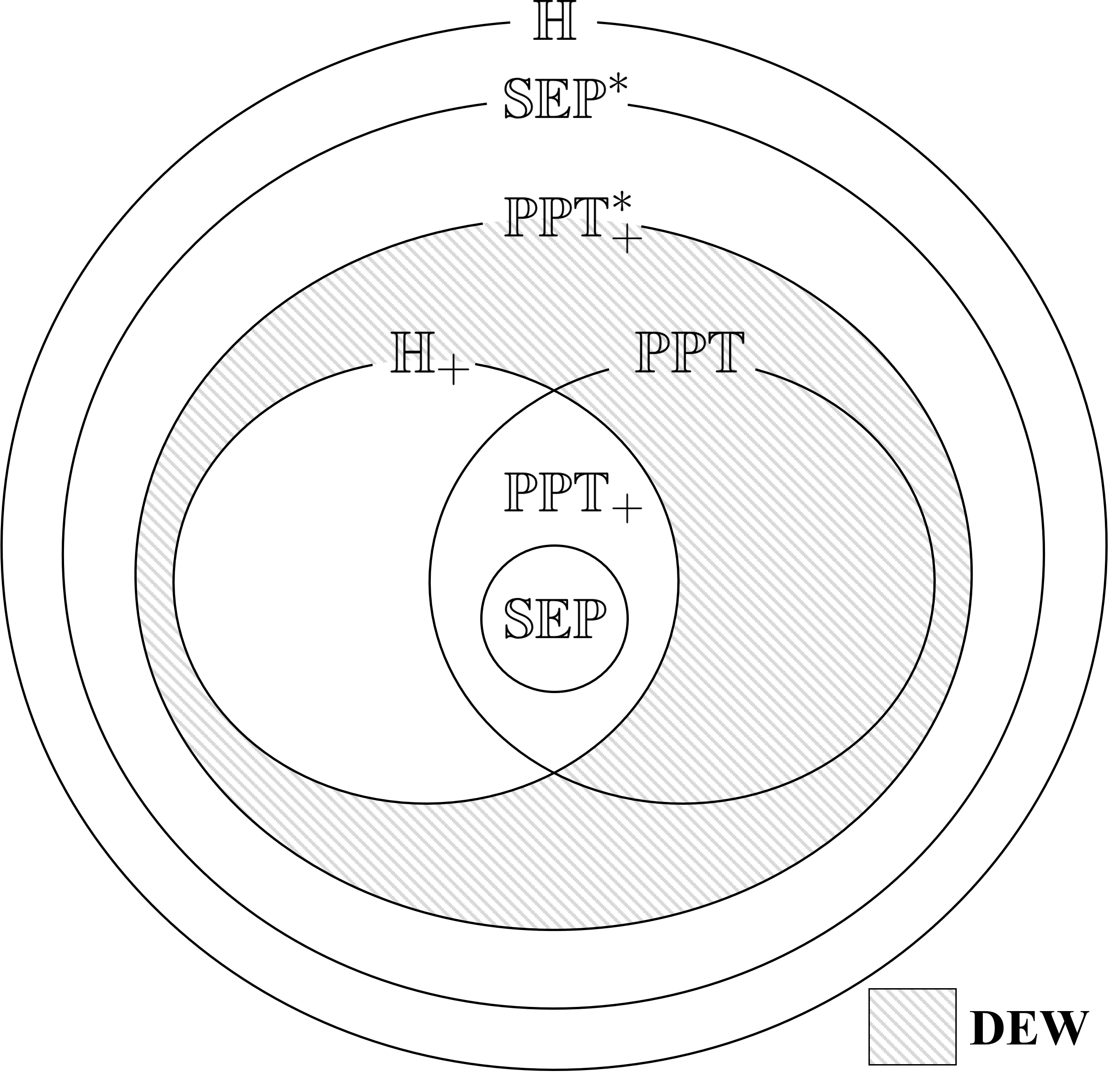}}
\caption{The relationship of the subsets of $\mathbb{H}$. 
$\mathbb{H}_{+}$ is the set of all positive-semidefinite operators and
$\mathbb{PPT}$ is the set of all PPT operators.
$\mathbb{PPT}_{+}$ is the intersection of $\mathbb{H}_{+}$ and $\mathbb{PPT}$.
The shaded area $\{W\in\mathbb{PPT}_{+}^{*}\,|\,W\notin\mathbb{H}_{+}\}$ is the set of all DEWs.
}\label{fig:inc}
\end{figure}

\indent From Proposition~\ref{pro:wptp}, we can see that
$W\in\mathbb{H}$ is a DEW if and only if
\begin{equation}
W\in\mathbb{PPT}_{+}^{*},~W\not\succeq0.
\end{equation}
We note that $\mathbb{PPT}_{+}^{*}\subseteq\mathbb{SEP}^{*}$ since $\mathbb{SEP}\subseteq\mathbb{PPT}_{+}$.
Figure~\ref{fig:inc} illustrates the relationship of the subsets of $\mathbb{H}$.

\subsection{Minimum-error discrimination of bipartite quantum states}
For a \emph{bipartite} quantum state ensemble,
\begin{equation}\label{eq:ens}
\mathcal{E}=\{\eta_{i},\rho_{i}\}_{i=1}^{n},
\end{equation}
where the state $\rho_{i}$ is prepared with the probability $\eta_{i}\in[0,1]$, let us consider the situation of discriminating the states from $\mathcal{E}$ using a measurement $\{M_{i}\}_{i=1}^{n}$. Here, the detection of $M_{i}$ means that 
we guess the prepared state as $\rho_{i}$.\\
\indent The \emph{minimum-error discrimination} of $\mathcal{E}$ is to achieve the optimal success probability,
\begin{equation}\label{eq:pgdef}
p_{\rm G}(\mathcal{E})=\max_{\rm measurement}\sum_{i=1}^{n}\eta_{i}\Tr(\rho_{i}M_{i}),
\end{equation}
where the maximum is taken over all possible measurements \cite{hels1969}.
We note that a measurement $\{M_{i}\}_{i=1}^{n}$ gives the optimal success probability $p_{\rm G}(\mathcal{E})$ if and only if it satisfies the following condition \cite{hole1974,yuen1975,barn20092,bae2013}:
\begin{equation}
\sum_{j=1}^{n}\eta_{j}\rho_{j}M_{j}-\eta_{i}\rho_{i}\succeq0~\forall i=1,\ldots,n.\label{eq:nscfme}
\end{equation}
\indent When the available measurements are limited to PPT measurements, 
we denote the maximum success probability by
\begin{equation}\label{eq:pptdef}
p_{\PPT}(\mathcal{E})=\max_{\substack{\rm PPT\\ \rm measurement}}\sum_{i=1}^{n}\eta_{i}\Tr(\rho_{i}M_{i}),
\end{equation}
where the maximum is taken over all possible PPT measurements.
Similarly, we denote 
\begin{equation}\label{eq:pldef}
p_{\rm L}(\mathcal{E})=\max_{\substack{\rm LOCC\\ \rm measurement}}\sum_{i=1}^{n}\eta_{i}\Tr(\rho_{i}M_{i}),
\end{equation}
where the maximum is taken over all possible LOCC measurements.\\
\indent Because $p_{\rm PPT}(\mathcal{E})$ is the maximum value over the set of all PPT measurements, which is a proper subset of the set of all measurements, $p_{\rm PPT}(\mathcal{E})$ is a lower bound of $p_{\rm G}(\mathcal{E})$.
Moreover, $p_{\rm L}(\mathcal{E})$ is a lower bound of $p_{\rm PPT}(\mathcal{E})$
since every LOCC measurement is a PPT measurement \cite{chit20142}. Thus, we have
\begin{equation}\label{eq:plptpg}
p_{\rm L}(\mathcal{E})\leqslant p_{\rm PPT}(\mathcal{E})\leqslant p_{\rm G}(\mathcal{E}).
\end{equation}
As $p_{\rm L}(\mathcal{E})$ and $p_{\PPT}(\mathcal{E})$ have the same objective function to maximize,
the first inequality in \eqref{eq:plptpg} becomes an equality
if and only if there exists
a LOCC measurement realizing $p_{\rm PPT}(\mathcal{E})$.
Similarly, 
the second inequality in \eqref{eq:plptpg} becomes an equality
if and only if there exists
a PPT measurement realizing $p_{\rm G}(\mathcal{E})$.

\section{PPT measurements and quantum state discrimination}\label{sec:ppq}
\indent In this section, we provide the first main result of our paper. We first consider an upper bound of $p_{\PPT}(\mathcal{E})$ in Eq.~\eqref{eq:pptdef}, and show that this upper bound is equal to $p_{\PPT}(\mathcal{E})$.
We further provide conditions on the maximum success probability over all possible PPT measurements.
The results will be used to obtain necessary and/or sufficient conditions for minimum-error discrimination by PPT measurements in the next section.\\
\indent For a bipartite quantum state ensemble $\mathcal{E}$ in Eq.~\eqref{eq:ens}, we define $\mathbb{H}_{\PPT}(\mathcal{E})$ as
\begin{equation}\label{eq:hpts}
\mathbb{H}_{\PPT}(\mathcal{E})=\{H\in\mathbb{H}\,|\, 
H-\eta_{i}\rho_{i}\in\mathbb{PPT}_{+}^{*}
~\mbox{for any}~i\in\{1,\ldots,n\}\},
\end{equation}
where $\mathbb{PPT}_{+}^{*}$ is defined in Eq.~\eqref{eq:pptpsd}.
In other words, $\mathbb{H}_{\PPT}(\mathcal{E})$ is the set of all $H$ such that $H-\eta_{i}\rho_{i}$ is decomposable in $\mathbb{H}$, for all $i=1,\ldots,n$. 
We further define
\begin{equation}\label{eq:hppte}
\mathbb{H}_{\DEW}(\mathcal{E})=\{H\in\mathbb{H}_{\PPT}(\mathcal{E})\,|\, H-\eta_{j}\rho_{j}
~\mbox{is a DEW}~\mbox{for some}~j\in\{1,\ldots,n\}\},
\end{equation}
that is, a subset of $\mathbb{H}_{\PPT}(\mathcal{E})$ satisfying $H-\eta_{j}\rho_{j}\not\succeq0$ for some $j\in\{1,\ldots,n\}$.
From the argument after Definition~\ref{def:ew},
we can see that 
\begin{equation}\label{eq:psmd}
H\in\mathbb{H}_{\PPT}(\mathcal{E})\setminus\mathbb{H}_{\DEW}(\mathcal{E})
\end{equation}
if and only if
\begin{equation}\label{eq:eqhpt}
H-\eta_{i}\rho_{i}\succeq0~\forall i=1,\ldots,n.
\end{equation}
\indent Now, let us consider the minimum quantity
\begin{equation}\label{eq:qptdef}
q_{\PPT}(\mathcal{E})=\min_{H\in\mathbb{H}_{\PPT}(\mathcal{E})}\Tr H,
\end{equation}
which is an upper bound of $p_{\rm PPT}(\mathcal{E})$\cite{cose2013}, that is,
\begin{equation}\label{eq:pptq}
p_{\rm PPT}(\mathcal{E})\leqslant q_{\PPT}(\mathcal{E}).
\end{equation}
The following theorem shows that 
$p_{\rm PPT}(\mathcal{E})$ in Eq.~\eqref{eq:pptdef}
is equal to 
$q_{\PPT}(\mathcal{E})$ in Eq.~\eqref{eq:qptdef}. The proof of Theorem~\ref{thm:pptq} is given in Appendix~\ref{app:thm1}.

\begin{theorem}\label{thm:pptq}
For a bipartite quantum state ensemble $\mathcal{E}=\{\eta_{i},\rho_{i}\}_{i=1}^{n}$, 
\begin{equation}\label{eq:ubppt}
p_{\rm PPT}(\mathcal{E})= q_{\PPT}(\mathcal{E}).
\end{equation}
\end{theorem}

\indent For a given ensemble $\mathcal{E}=\{\eta_{i},\rho_{i}\}_{i=1}^{n}$, the following theorem provides a necessary and sufficient condition for a PPT measurement $\{M_{i}\}_{i=1}^{n}$ and $H\in\mathbb{H}_{\PPT}(\mathcal{E})$ to realize $p_{\PPT}(\mathcal{E})$ and $q_{\PPT}(\mathcal{E})$, respectively.

\begin{theorem}\label{thm:mnsc}
For a bipartite quantum state ensemble $\mathcal{E}=\{\eta_{i},\rho_{i}\}_{i=1}^{n}$, a PPT measurement $\{M_{i}\}_{i=1}^{n}$ and $H\in\mathbb{H}_{\PPT}(\mathcal{E})$, $\{M_{i}\}_{i=1}^{n}$ realizes $p_{\PPT}(\mathcal{E})$ and $H$ provides $q_{\PPT}(\mathcal{E})$ if and only if
\begin{equation}\label{eq:comc}
\Tr[M_{i}(H-\eta_{i}\rho_{i})]=0~~\forall i=1,\ldots,n.
\end{equation}
\end{theorem}
\begin{proof}
Let us suppose that $\{M\}_{i=1}^{n}$ and $H$ 
give $p_{\PPT}(\mathcal{E})$ and $q_{\PPT}(\mathcal{E})$, respectively.
From $M_{i}\in\mathbb{PPT}_{+}$ and $H-\eta_{i}\rho_{i}\in\mathbb{PPT}_{+}^{*}$ for all $i=1,\ldots,n$, we have
\begin{equation}\label{eq:mihige}
\Tr[M_{i}(H-\eta_{i}\rho_{i})]\geqslant0~\forall i=1,\ldots,n.
\end{equation}
\indent We note that
\begin{equation}\label{eq:stmh}
\sum_{i=1}^{n}\Tr[M_{i}(H-\eta_{i}\rho_{i})]
=\Tr H-\sum_{i=1}^{n}\eta_{i}\Tr(\rho_{i}M_{i})
=q_{\PPT}(\mathcal{E})-p_{\PPT}(\mathcal{E})=0,
\end{equation}
where the first equality is from $\sum_{i=1}^{n}M_{i}=\mathbbm{1}$,
the second equality is due to the assumption of $H$ and $\{M_{i}\}_{i=1}^{n}$, 
and the last equality is by Theorem~\ref{thm:pptq}. 
Inequality~\eqref{eq:mihige} and Eq.~\eqref{eq:stmh}
lead us to Condition~\eqref{eq:comc}.\\
\indent Conversely let us assume that 
$\{M_{i}\}_{i=1}^{n}$ and $H$ satisfy Condition~\eqref{eq:comc}.
This assumption implies
\begin{equation}\label{eq:qpthq}
q_{\PPT}(\mathcal{E})
=p_{\PPT}(\mathcal{E})
\geqslant\sum_{i=1}^{n}\eta_{i}\mathrm{Tr}(\rho_{i}M_{i})
=\sum_{i=1}^{n}\eta_{i}\mathrm{Tr}(\rho_{i}M_{i})
+\sum_{i=1}^{n}\mathrm{Tr}[M_{i}(H-\eta_{i}\rho_{i})]
=\mathrm{Tr}H\geqslant q_{\PPT}(\mathcal{E}),
\end{equation}
where the first equality follows from Theorem~\ref{thm:pptq},
the second equality is from Condition~\eqref{eq:comc},
the last equality is due to $\sum_{i=1}^{n}M_{i}=\mathbbm{1}$, 
and the first and second inequalities are from 
the definitions of $p_{\PPT}(\mathcal{E})$ and $q_{\PPT}(\mathcal{E})$, respectively.
Inequality~\eqref{eq:qpthq} leads us to 
\begin{equation}
\sum_{i=1}^{n}\eta_{i}\mathrm{Tr}(\rho_{i}M_{i})=p_{\PPT}(\mathcal{E}),~~
\Tr H=q_{\PPT}(\mathcal{E}). 
\end{equation}
Thus, $\{M_{i}\}_{i=1}^{n}$ and $H$ give $p_{\PPT}(\mathcal{E})$ and $q_{\PPT}(\mathcal{E})$, respectively.
\end{proof}
\indent We note that $H\in\mathbb{H}_{\PPT}(\mathcal{E})$ giving $q_{\PPT}(\mathcal{E})$ is generally not unique (see Example~\ref{ex:hnu} in Section~\ref{subsec:nsc}). 
However, the following corollary states the case that 
$H\in\mathbb{H}_{\PPT}(\mathcal{E})$ providing $q_{\PPT}(\mathcal{E})$ is unique.

\begin{corollary}\label{cor:exer1}
For a bipartite quantum state ensemble $\mathcal{E}=\{\eta_{i},\rho_{i}\}_{i=1}^{n}$, we have
\begin{equation}\label{eq:ppte1}
p_{\PPT}(\mathcal{E})=\eta_{1},
\end{equation}
if and only if 
\begin{equation}\label{eq:exppt}
\eta_{1}\rho_{1}-\eta_{i}\rho_{i}\in\mathbb{PPT}_{+}^{*}~~\forall i=2,\ldots,n.
\end{equation}
In this case, $\eta_{1}\rho_{1}$ is the only element of $\mathbb{H}_{\PPT}(\mathcal{E})$ providing $q_{\PPT}(\mathcal{E})$.
\end{corollary}
\begin{proof}
Let $\{M_{i}\}_{i=1}^{n}$ be the measurement that $M_{1}=\mathbbm{1}$ and $M_{2},\ldots,M_{n}$ are the zero operator in $\mathbb{H}$. 
We first assume Eq.~\eqref{eq:ppte1} and
consider
$H\in\mathbb{H}_{\PPT}(\mathcal{E})$ providing $q_{\PPT}(\mathcal{E})$.
Since $\{M_{i}\}_{i=1}^{n}$ is obviously a PPT measurement giving $p_{\PPT}(\mathcal{E})$,
it follows from Theorem~\ref{thm:mnsc} that 
$\Tr(H-\eta_{1}\rho_{1})=0$.
From $H-\eta_{1}\rho_{1}\in\mathbb{PPT}_{+}^{*}$ and Lemma~\ref{lem:ppti} in Appendix~\ref{app:thm1}, we have $H=\eta_{1}\rho_{1}$.
Thus, $H\in\mathbb{H}_{\PPT}(\mathcal{E})$ 
together with the definition of $\mathbb{H}_{\PPT}(\mathcal{E})$ 
leads us to Condition~\eqref{eq:exppt}.\\
\indent Conversely, let us suppose Condition~\eqref{eq:exppt} and
consider $H=\eta_{1}\rho_{1}$. Condition~\eqref{eq:exppt} implies $H\in\mathbb{H}_{\PPT}(\mathcal{E})$.
The PPT measurement $\{M_{i}\}_{i=1}^{n}$ and $H\in\mathbb{H}_{\PPT}(\mathcal{E})$ satisfy Condition~\eqref{eq:comc}. Therefore, we have
\begin{equation}
p_{\PPT}(\mathcal{E})=q_{\PPT}(\mathcal{E})=\Tr H=\eta_{1},
\end{equation}
where the first equality is from Theorem~\ref{thm:pptq}
and the second equality follows from Theorem~\ref{thm:mnsc}.
\end{proof}
\indent When Eq.~\eqref{eq:ppte1} of Corollary~\ref{cor:exer1} holds, the maximum success probability $p_{\PPT}(\mathcal{E})$ can be achieved without the help of measurement,
simply by guessing $\rho_{1}$ is prepared.
As we can check in the proof of Corollary~\ref{cor:exer1}, the choice of $\rho_{1}$ in Corollary~\ref{cor:exer1} can be arbitrary. That is, any of $\{\rho_{i}\}_{i=1}^{n}$ can be used to play the role of $\rho_{1}$ in Corollary~\ref{cor:exer1}.

\section{Minimum-error discrimination by PPT measurements}\label{sec:mep}
\indent In this section, we provide another main result of our paper showing that the minimum-error discrimination of bipartite quantum states using PPT measurements strongly depends on the existence of DEW. More precisely, we establish necessary and/or sufficient conditions for minimum-error discrimination by PPT measurements in terms of DEW.\\
\indent For a bipartite quantum state ensemble $\mathcal{E}$ in Eq.~\eqref{eq:ens}, the minimum-error discrimination can be realized by PPT measurements if and only if
\begin{equation}\label{eq:inptpg}
p_{\PPT}(\mathcal{E})=p_{\rm G}(\mathcal{E}),
\end{equation}
where $p_{\rm G}(\mathcal{E})$ and $p_{\PPT}(\mathcal{E})$ are defined in Eqs.~\eqref{eq:pgdef} and \eqref{eq:pptdef}, respectively. 
Here, we provide necessary and/or sufficient conditions for Eq.~\eqref{eq:inptpg} in terms of DEW.

\subsection{Necessary condition for $p_{\PPT}(\mathcal{E})=p_{\rm G}(\mathcal{E})$}

\begin{theorem}\label{thm:qmsc}
For a bipartite quantum state ensemble $\mathcal{E}=\{\eta_{i},\rho_{i}\}_{i=1}^{n}$,
if $p_{\PPT}(\mathcal{E})=p_{\rm G}(\mathcal{E})$, then
there does not exist PPT measurement $\{M_{i}\}_{i=1}^{n}$ satisfying
\begin{equation}\label{eq:dewc}
\sum_{i=1}^{n}\eta_{i}\rho_{i}M_{i}\in\mathbb{H}_{\DEW}(\mathcal{E}),
\end{equation}
where $\mathbb{H}_{\DEW}(\mathcal{E})$ is defined in Eq.~\eqref{eq:hppte}.
Moreover, if $\{M_{i}\}_{i=1}^{n}$ is a PPT measurement satisfying Condition~\eqref{eq:dewc}, we have
\begin{equation}\label{eq:ppes}
p_{\PPT}(\mathcal{E})=\sum_{i=1}^{n}\eta_{i}\Tr(\rho_{i}M_{i}).
\end{equation}
\end{theorem}
\begin{proof}
Let us suppose $\{M_{i}\}_{i=1}^{n}$ is a PPT measurement satisfying Condition~\eqref{eq:dewc}, and consider
\begin{equation}
H=\sum_{i=1}^{n}\eta_{i}\rho_{i}M_{i}.
\end{equation}
Equation~\eqref{eq:ppes} holds because
\begin{equation}
\sum_{i=1}^{n}\eta_{i}\Tr(\rho_{i}M_{i})
\leqslant p_{\PPT}(\mathcal{E})=q_{\PPT}(\mathcal{E})
\leqslant \Tr H=\sum_{i=1}^{n}\eta_{i}\Tr(\rho_{i}M_{i}),
\end{equation}
where the first and second inequalities follow from the definitions of $p_{\PPT}(\mathcal{E})$ and $q_{\PPT}(\mathcal{E})$, respectively,
the first equality is due to Theorem~\ref{thm:pptq}, and the second equality is from the definition of $H$. This proves the second statement of our theorem.\\
\indent Now, assume $p_{\PPT}(\mathcal{E})=p_{\rm G}(\mathcal{E})$.
If there exists a PPT measurement $\{M_{i}'\}_{i=1}^{n}$ satisfying Condition~\eqref{eq:dewc},
the second statement of our theorem implies that $\{M_{i}'\}_{i=1}^{n}$
gives the optimal success probability $p_{\rm G}(\mathcal{E})$ due to Eq.~\eqref{eq:ppes}.
From the optimality condition in Eq.~\eqref{eq:nscfme},
we have
\begin{equation}
\sum_{j=1}^{n}\eta_{j}\rho_{j}M_{j}'-\eta_{i}\rho_{i}
\succeq0~~\forall i=1,\ldots,n, 
\end{equation}
which contradicts Condition~\eqref{eq:dewc}.
Thus, there does not exist PPT measurement satisfying Condition~\eqref{eq:dewc}.
This proves the first statement of our theorem.
\end{proof}
\begin{example}\label{ex:eonc}
For any integer $d\geqslant 2$, let us consider the two-qudit state ensemble $\mathcal{E}=\{\eta_{i,j}^{(k,l)},\rho_{i,j}^{(k,l)}\}_{i,j,k,l}$ consisting of $2d(d-1)$ mixed states with equal prior probabilities,
\begin{alignat}{3}
&\eta_{i,j}^{(k,l)}=\frac{1}{2d(d-1)},~
\rho_{i,j}^{(k,l)}=\frac{\lambda}{3}(\Psi_{i,j}^{(k)}+\Pi_{i,j}^{(l)})+(1-\lambda)\sigma,\nonumber\\
&i,j\in\{0,1,\ldots,d-1\}~\mbox{with}~i<j,~~k,l\in\{1,2\},\label{eq:exes}
\end{alignat}
where $\sigma$ is an arbitrary two-qudit state, $0<\lambda\leqslant 1$, and
\begin{alignat}{3}
\Psi_{i,j}^{(1)}=\frac{1}{2}(\ket{ii}+\ket{jj})(\bra{ii}+\bra{jj}),&~
\Pi_{i,j}^{(1)}=\ketbra{ii}{ii}+\ketbra{jj}{jj},\nonumber\\
\Psi_{i,j}^{(2)}=\frac{1}{2}(\ket{ii}-\ket{jj})(\bra{ii}-\bra{jj}),&~
\Pi_{i,j}^{(2)}=\ketbra{ij}{ij}+\ketbra{ji}{ji}.
\label{eq:psijk}
\end{alignat}
\end{example}
For a PPT measurement $\{M_{i,j}^{(k,l)}\}_{i,j,k,l}$ with
\begin{equation}\label{eq:mij}
M_{i,j}^{(1,1)}=M_{i,j}^{(2,1)}=\frac{1}{2(d-1)}\Pi_{i,j}^{(1)},~
M_{i,j}^{(1,2)}=M_{i,j}^{(2,2)}=\frac{1}{2}\Pi_{i,j}^{(2)},
\end{equation}
we show that Condition~\eqref{eq:dewc} holds with respect to the ensemble in Eq.~\eqref{eq:exes} and the measurement in Eq.~\eqref{eq:mij}.
It is straightforward to verify that
\begin{alignat}{3}
\Big[\sum_{i',j',k',l'}\eta_{i',j'}^{(k',l')}\rho_{i',j'}^{(k',l')}M_{i',j'}^{(k',l')}-\eta_{i,j}^{(k,1)}\rho_{i,j}^{(k,1)}\Big]^{\PT}
&=\frac{\lambda}{4d(d-1)}\Big[\Psi_{i,j}^{(5-k)}
+\frac{1}{3}\Psi_{i,j}^{(k+2)}
+\hat{\Pi}_{i,j}^{(1)}
+\frac{2}{3}\hat{\Pi}_{i,j}^{(2)}\Big]\succeq0,
\nonumber\\
\sum_{i',j',k',l'}\eta_{i',j'}^{(k',l')}\rho_{i',j'}^{(k',l')}M_{i',j'}^{(k',l')}-\eta_{i,j}^{(k,2)}\rho_{i,j}^{(k,2)}
&=\frac{\lambda}{4d(d-1)}
\Big[\Psi_{i,j}^{(3-k)}+\frac{1}{3}\Psi_{i,j}^{(k)}
+\hat{\Pi}_{i,j}^{(1)}+\frac{2}{3}\hat{\Pi}_{i,j}^{(2)}\Big]\succeq0,\nonumber\\
&\forall i,j\in\{0,1,\ldots,d-1\}~\mbox{with}~i<j,~~\forall k\in\{1,2\},\label{eq:sclt1}
\end{alignat}
where
\begin{alignat}{3}
\Psi_{i,j}^{(3)}=\frac{1}{2}(\ket{ij}+\ket{ji})(\bra{ij}+\bra{ji}),~
&\hat{\Pi}_{i,j}^{(1)}=\sum_{i'=0}^{d-1}\ketbra{i'i'}{i'i'}-\Pi_{i,j}^{(1)},
\nonumber\\
\Psi_{i,j}^{(4)}=\frac{1}{2}(\ket{ij}-\ket{ji})(\bra{ij}-\bra{ji}),~
&\hat{\Pi}_{i,j}^{(2)}=\sum_{\substack{i',j'=0\\i'\neq j'}}^{d-1}\ketbra{i'j'}{i'j'}-\Pi_{i,j}^{(2)}.
\label{eq:gbes}
\end{alignat}
\indent From the positivity in \eqref{eq:sclt1}, we have
\begin{alignat}{3}
&\sum_{i',j',k',l'}\eta_{i',j'}^{(k',l')}\rho_{i',j'}^{(k',l')}M_{i',j'}^{(k',l')}
-\eta_{i,j}^{(k,l)}\rho_{i,j}^{(k,l)}\in\mathbb{PPT}_{+}^{*}\nonumber\\
&\forall i,j\in\{0,1,\ldots,d-1\}~\mbox{with}~i<j,~~\forall k,l\in\{1,2\}.
\label{eq:ipml}
\end{alignat}
Furthermore, a straightforward calculation leads us to
\begin{alignat}{3}
&\Tr\Big[\Big(\sum_{i',j',k',l'}\eta_{i',j'}^{(k',l')}\rho_{i',j'}^{(k',l')}M_{i',j'}^{(k',l')}-\eta_{i,j}^{(k,1)}\rho_{i,j}^{(k,1)}\Big)\Psi_{i,j}^{(k)}\Big]
=-\frac{\lambda}{12d(d-1)}\nonumber\\
&\forall i,j\in\{0,1,\ldots,d-1\}~\mbox{with}~i<j,~~\forall k\in\{1,2\}.\label{eq:mwhe}
\end{alignat}
From Eq.~\eqref{eq:mwhe}, we have
\begin{alignat}{3}
&\sum_{i',j',k',l'}\eta_{i',j'}^{(k',l')}\rho_{i',j'}^{(k',l')}M_{i',j'}^{(k',l')}
-\eta_{i,j}^{(k,1)}\rho_{i,j}^{(k,1)}\not\succeq0\nonumber\\
&\forall i,j\in\{0,1,\ldots,d-1\}~\mbox{with}~i<j,~~\forall k\in\{1,2\}.
\label{eq:mkle}
\end{alignat}
From Eqs.~\eqref{eq:ipml} and \eqref{eq:mkle},
the ensemble in Eq.~\eqref{eq:exes} and the measurement in Eq.~\eqref{eq:mij} satisfy
Condition~\eqref{eq:dewc}.\\
\indent The second statement of Theorem~\ref{thm:qmsc} implies that $p_{\PPT}(\mathcal{E})$ is the expectation of the measurement in Eq.~\eqref{eq:mij} about the ensemble in Eq.~\eqref{eq:exes}:
\begin{equation}
p_{\PPT}(\mathcal{E})
=\frac{6+\lambda(2d-3)(d+2)}{12d(d-1)}.
\end{equation}
From the first statement of Theorem~\ref{thm:qmsc}, we have
\begin{equation}
p_{\PPT}(\mathcal{E})<p_{\rm G}(\mathcal{E}).
\end{equation}
\indent We also note that the measurement in Eq.~\eqref{eq:mij} is a LOCC measurement because it can be implemented by performing the same local measurement $\{\ketbra{l}{l}\}_{l=0}^{d-1}$ on each party. Thus, we have
\begin{equation}
p_{\rm L}(\mathcal{E})=p_{\PPT}(\mathcal{E})
=\frac{6+\lambda(2d-3)(d+2)}{12d(d-1)}.
\end{equation}

\subsection{Necessary and sufficient condition for $p_{\PPT}(\mathcal{E})=p_{\rm G}(\mathcal{E})$}\label{subsec:nsc}

\begin{theorem}\label{thm:qupb}
For a bipartite quantum state ensemble $\mathcal{E}=\{\eta_{i},\rho_{i}\}_{i=1}^{n}$, $p_{\PPT}(\mathcal{E})=p_{\rm G}(\mathcal{E})$ if and only if
there exists $H\in\mathbb{H}_{\PPT}(\mathcal{E})$
such that it provides $q_{\PPT}(\mathcal{E})$
but does not satisfy
\begin{equation}\label{eq:hdc}
H\in\mathbb{H}_{\DEW}(\mathcal{E}),
\end{equation}
or equivalently, there is $H\in\mathbb{H}$ satisfying Condition~\eqref{eq:psmd} and $\Tr H=q_{\PPT}(\mathcal{E})$.
\end{theorem}
\begin{proof}
Let $\{M_{i}\}_{i=1}^{n}$ be a PPT measurement providing $p_{\PPT}(\mathcal{E})$. We first suppose $p_{\PPT}(\mathcal{E})=p_{\rm G}(\mathcal{E})$ and consider 
\begin{equation}\label{eq:hjn}
H=\sum_{i=1}^{n}\eta_{i}\rho_{i}M_{i}.
\end{equation}
Since the measurement $\{M_{i}\}_{i=1}^{n}$ gives
the optimal success probability $p_{\rm G}(\mathcal{E})$ due to the assumption of $\{M_{i}\}_{i=1}^{n}$ and $p_{\PPT}(\mathcal{E})=p_{\rm G}(\mathcal{E})$,
we have from the optimality condition in Eq.~\eqref{eq:nscfme} that
\begin{alignat}{3}
H-\eta_{i}\rho_{i}
=\sum_{j=1}^{n}\eta_{j}\rho_{j}M_{j}-\eta_{i}\rho_{i}\succeq0~~\forall i=1,\ldots,n.
\label{eq:hemp}
\end{alignat}
Therefore, $H$ satisfies Condition~\eqref{eq:psmd}. Moreover, we have
\begin{equation}\label{eq:trqpt}
\Tr H=\sum_{i=1}^{n}\eta_{i}\Tr(\rho_{i}M_{i})=p_{\PPT}(\mathcal{E})=q_{\PPT}(\mathcal{E}),
\end{equation}
where the first equality is from Eq.~\eqref{eq:hjn},
the second equality follows from the assumption of $\{M_{i}\}_{i=1}^{n}$, and 
the last equality is due to Theorem~\ref{thm:pptq}.\\
\indent Conversely, let us assume $H$ is an element of $\mathbb{H}$ satisfying Condition~\eqref{eq:psmd} and $\Tr H=q_{\PPT}(\mathcal{E})$.
Thus, the positivity in \eqref{eq:eqhpt} is satisfied in terms of $H$.
For each $i=1,\ldots,n$, the positive-semidefinite operators $M_{i}$ and $H-\eta_{i}\rho_{i}$ are orthogonal since they satisfy Condition~\eqref{eq:comc} from Theorem~\ref{thm:mnsc}.\\
\indent The optimality condition in Eq.~\eqref{eq:nscfme} holds for the measurement $\{M_{i}\}_{i=1}^{n}$ because
\begin{equation}\label{eq:hetai}
\sum_{j=1}^{n}\eta_{j}\rho_{j}M_{j}-\eta_{i}\rho_{i}
=\sum_{j=1}^{n}\eta_{j}\rho_{j}M_{j}
+\sum_{k=1}^{n}(H-\eta_{k}\rho_{k})M_{k}-\eta_{i}\rho_{i}
=H-\eta_{i}\rho_{i}\succeq0~~\forall i=1,\ldots,n,
\end{equation}
where the first equality is from the orthogonality of $M_{i}$ and $H-\eta_{i}\rho_{i}$ for each $i=1,\ldots,n$ and the second equality is from $\sum_{i=1}^{n}M_{i}=\mathbbm{1}$.
Thus, we have
\begin{equation}
p_{\rm G}(\mathcal{E})
=\sum_{i=1}^{n}\eta_{i}\Tr(\rho_{i}M_{i})
=p_{\PPT}(\mathcal{E}),
\end{equation}
where the second equality is due to the assumption of $\{M_{i}\}_{i=1}^{n}$.
\end{proof}

\indent If $p_{\PPT}(\mathcal{E})=p_{\rm G}(\mathcal{E})$,
Theorem~\ref{thm:qupb} implies that there must exist $H\in\mathbb{H}_{\PPT}(\mathcal{E})\setminus\mathbb{H}_{\DEW}(\mathcal{E})$ giving $q_{\PPT}(\mathcal{E})$.
In this case, there possibly exists another Hermitian operator $H'$ satisfying
$H'\in\mathbb{H}_{\DEW}(\mathcal{E})$ and $\Tr H'= q_{\PPT}(\mathcal{E})$, 
which is illustrated in the following example.
\begin{example}\label{ex:hnu}
Let us consider the two-qubit state ensemble $\mathcal{E}=\{\eta_{i},\rho_{i}\}_{i=1}^{3}$ consisting of three pure orthogonal states with equal prior probabilities,
\begin{alignat}{3}
&\eta_{1}=\frac{1}{3},~\rho_{1}=\ketbra{0}{0}\otimes\ketbra{0}{0},\nonumber\\
&\eta_{2}=\frac{1}{3},~\rho_{2}=\ketbra{1}{1}\otimes\ketbra{1}{1},\nonumber\\
&\eta_{3}=\frac{1}{3},~\rho_{3}=\Psi_{+},\label{eq:tpose}
\end{alignat}
where 
\begin{equation}\label{eq:psipm}
\Psi_{\pm}=\ketbra{\Psi_{\pm}}{\Psi_{\pm}},~
\ket{\Psi_{\pm}}=\frac{1}{\sqrt{2}}(\ket{01}\pm\ket{10}).
\end{equation}
\end{example}
\indent Since the states $\rho_{1}$ and $\rho_{2}$ are product states, we have
\begin{equation}\label{eq:expl2}
p_{\rm L}(\mathcal{E})=1.
\end{equation}
We note that three $2\otimes2$ pure orthogonal states 
can be perfectly discriminated by a finite-round LOCC 
if and only if at least two of them are product states\cite{walg2002}.
Since $p_{\rm G}(\mathcal{E})$ is bounded above by 1, 
it follows from Inequality~\eqref{eq:plptpg} and Eq.~\eqref{eq:expl2} that
\begin{equation}\label{eq:allpeq}
p_{\rm L}(\mathcal{E})=p_{\PPT}(\mathcal{E})=p_{\rm G}(\mathcal{E})=1.
\end{equation}
Furthermore, we have
\begin{equation}\label{eq:qex2}
q_{\PPT}(\mathcal{E})=p_{\PPT}(\mathcal{E})=1,
\end{equation}
where the first equality follows from Theorem~\ref{thm:pptq}
and the second equality is due to Eq.~\eqref{eq:allpeq}.\\
\indent Now, let us consider the Hermitian operator
\begin{equation}\label{eq:hex2}
H=\frac{1}{3}(\Phi_{+}+\Phi_{-})+\frac{1+t}{6}\Psi_{+}+\frac{1-t}{6}\Psi_{-},
\end{equation}
where $\Psi_{\pm}$ is defined in Eq.~\eqref{eq:psipm}, $0\leqslant t\leqslant 1$, and
\begin{equation}
\Phi_{\pm}=\ketbra{\Phi_{\pm}}{\Phi_{\pm}},~
\ket{\Phi_{\pm}}=\frac{1}{\sqrt{2}}(\ket{00}\pm\ket{11}).
\end{equation}
Equations~\eqref{eq:qex2} and \eqref{eq:hex2} imply
\begin{equation}\label{eq:trh1}
\Tr H=1=q_{\PPT}(\mathcal{E}).
\end{equation}
A straightforward calculation leads us to
\begin{alignat}{3}
&H-\eta_{1}\rho_{1}=
\frac{1}{3}\rho_{2}+\frac{1+t}{6}\Psi_{+}+\frac{1-t}{6}\Psi_{-}\succeq0,
\nonumber\\
&H-\eta_{2}\rho_{2}=
\frac{1}{3}\rho_{1}+\frac{1+t}{6}\Psi_{+}+\frac{1-t}{6}\Psi_{-}\succeq0,
\label{eq:exhe}\\
&H-\eta_{3}\rho_{3}
=\frac{1+t}{6}(\Phi_{+}+\Phi_{-})+\frac{1-t}{3}\Phi_{-}^{\PT}\in\mathbb{PPT}_{+}^{*}.
\nonumber
\end{alignat}
Due to Eqs.~\eqref{eq:trh1} and \eqref{eq:exhe}, $H$ is an element of $\mathbb{H}_{\PPT}(\mathcal{E})$
giving $q_{\PPT}(\mathcal{E})$ regardless of $t\in[0,1]$.
Moreover, $H\in\mathbb{H}_{\DEW}(\mathcal{E})$ for all $t\in[0,1)$ because
\begin{equation}
\Tr[(H-\eta_{3}\rho_{3})\Psi_{+}]=-\frac{1-t}{6}.
\end{equation}
\indent For a bipartite quantum state ensemble $\mathcal{E}=\{\eta_{i},\rho_{i}\}_{i=1}^{n}$
where $H$ is the only element of $\mathbb{H}_{\PPT}(\mathcal{E})$ giving $q_{\PPT}(\mathcal{E})$,
Theorem~\ref{thm:qupb} tells us that $p_{\PPT}(\mathcal{E})=p_{\rm G}(\mathcal{E})$
if and only if 
there is no DEW in
$\{H-\eta_{i}\rho_{i}\}_{i=1}^{n}$.
From Corollary~\ref{cor:exer1},
$\eta_{1}\rho_{1}$ is the only element of $\mathbb{H}_{\PPT}(\mathcal{E})$ providing $q_{\PPT}(\mathcal{E})$ when Condition~\eqref{eq:exppt} holds.
Thus, we have the following corollary.

\begin{corollary}\label{cor:pql1}
For a bipartite quantum state ensemble $\mathcal{E}=\{\eta_{i},\rho_{i}\}_{i=1}^{n}$ with Condition~\eqref{eq:exppt}, $p_{\PPT}(\mathcal{E})=p_{\rm G}(\mathcal{E})$ if and only if 
there is no DEW in $\{\eta_{1}\rho_{1}-\eta_{i}\rho_{i}\}_{i=2}^{n}$.
\end{corollary}

\begin{example}\label{ex:eons}
For any integer $d\geqslant 2$, let us consider the two-qu$d$it state ensemble $\mathcal{E}=\{\eta_{1},\rho_{1}\}\cup\{\eta_{i,j}^{(k)},\rho_{i,j}^{(k)}\}_{i,j,k}$ consisting of $1+2d(d-1)$ states, 
\begin{alignat}{3}
&\eta_{1}=\frac{d}{5d-4},~\rho_{1}=\frac{1}{d^{2}}\mathbbm{1},\nonumber\\
&\eta_{i,j}^{(k)}=\frac{2}{d(5d-4)},~\rho_{i,j}^{(k)}=(1-\lambda)\Psi_{i,j}^{(k)}+\frac{\lambda}{d^{2}}\mathbbm{1},\nonumber\\
&i,j\in\{0,1,\ldots,d-1\}~\mbox{with}~i<j,~k=1,2,3,4,\label{eq:exerho}
\end{alignat}
where $0\leqslant\lambda<1$ and $\Psi_{i,j}^{(k)}$ is defined in Eqs.~\eqref{eq:psijk} and \eqref{eq:gbes}.
\end{example}
\indent From a straightforward calculation, we can verify that
\begin{eqnarray}
\eta_{1}\rho_{1}-\eta_{i,j}^{(k)}\rho_{i,j}^{(k)}
&=&\frac{1}{d^{3}(5d-4)}\Big[(d^{2}-2\lambda)(\mathbbm{1}-\mathbbm{1}_{i,j})
+(d^{2}-2\lambda)\mathbbm{1}_{i,j}
-2d^{2}(1-\lambda)\Psi_{i,j}^{(k)}\Big]\nonumber\\
&=&\frac{1}{d^{3}(5d-4)}\Big[(d^{2}-2\lambda)(\mathbbm{1}-\mathbbm{1}_{i,j})
+\lambda(d^{2}-2)\mathbbm{1}_{i,j}+2d^{2}(1-\lambda)\Psi_{i,j}^{(5-k)\PT}\Big]
\in\mathbb{PPT}_{+}^{*}\nonumber\\
&&\forall i,j\in\{0,1,\ldots,d-1\}~\mbox{with}~i<j,~~\forall k\in\{1,2,3,4\}.
\label{eq:erm0}
\end{eqnarray}
where 
\begin{equation}
\mathbbm{1}_{i,j}=(\ketbra{i}{i}+\ketbra{j}{j})\otimes(\ketbra{i}{i}+\ketbra{j}{j}),
\end{equation}
and the second equality in Eq.~\eqref{eq:erm0} is due to
\begin{alignat}{3}
\frac{1}{2}\mathbbm{1}_{i,j}-\Psi_{i,j}^{(k)}=\Psi_{i,j}^{(5-k)\PT}
~\forall i,j\in\{0,1,\ldots,d-1\}~\mbox{with}~i<j,~~\forall k\in\{1,2,3,4\}.
\end{alignat}
From Corollary~\ref{cor:exer1} and Eq.~\eqref{eq:erm0}, we have
\begin{equation}\label{eq:ped}
p_{\PPT}(\mathcal{E})=\eta_{1}=\frac{d}{5d-4}.
\end{equation}
\indent For each $i,j,k$, we can easily see from the first equality in Eq.~\eqref{eq:erm0} that
$\eta_{1}\rho_{1}-\eta_{i,j}^{(k)}\rho_{i,j}^{(k)}\succeq0$ if and only if
\begin{equation}\label{eq:inr}
\Tr\Big[\Big(\eta_{1}\rho_{1}-\eta_{i,j}^{(k)}\rho_{i,j}^{(k)}\Big)\Psi_{i,j}^{(k)}\Big]
=\frac{2\lambda d^{2}-2\lambda-d^{2}}{d^{3}(5d-4)}
\geqslant0.
\end{equation}
Therefore, it follows from Corollary~\ref{cor:pql1} that $p_{\PPT}(\mathcal{E})=p_{\rm G}(\mathcal{E})$ if and only if
\begin{equation}
\lambda\geqslant\frac{d^{2}}{2(d^{2}-1)}.
\end{equation}

\section{Discussion}\label{sec:dis}
We have considered bipartite quantum state discrimination and shown that the minimum-error discrimination of bipartite quantum states using PPT measurements strongly depends on the existence of DEW.
We have established conditions on minimum-error discrimination by PPT measurements, that is, $p_{\PPT}(\mathcal{E})=p_{\rm G}(\mathcal{E})$, in terms of DEW(Theorems~\ref{thm:qmsc} and \ref{thm:qupb}). 
We have also provided conditions on the maximum success probability over all possible PPT measurements(Theorems~\ref{thm:pptq} and \ref{thm:mnsc}). 
Our results have been illustrated by examples of multidimensional bipartite quantum states.\\
\indent Quantum nonlocality is a genuine phenomenon of multipartite quantum systems without any classical counterpart.
Quantum nonlocality is a key ingredient making quantum states outperform the classical ones in various quantum information processing tasks such as quantum teleportation and quantum cryptography\cite{eker1991,benn1993}. 
It is also known that this quantum nonlocality plays an important role in quantum algorithms which are more powerful than any classical ones\cite{deut1992,shor1994}. 
Quantum nonlocality also occurs in discriminating multipartite quantum states;
in discriminating states from a bipartite quantum state ensemble $\mathcal{E}$, quantum nonlocality occurs when optimal state discrimination cannot be realized only by LOCC measurement, that is, $p_{\rm L}(\mathcal{E})<p_{\rm G}(\mathcal{E})$.\\
\indent Our results here can be used to detect the occurrence of nonlocality in quantum state discrimination.
Violation of the condition in Theorem~\ref{thm:qmsc} or \ref{thm:qupb} means $p_{\PPT}(\mathcal{E})<p_{\rm G}(\mathcal{E})$,
which consequently implies $p_{\rm L}(\mathcal{E})<p_{\rm G}(\mathcal{E})$.
Thus, Condition~\eqref{eq:dewc} of Theorem~\ref{thm:qmsc} and
Condition~\eqref{eq:hdc} of Theorem~\ref{thm:qupb} can be used as
sufficient conditions on $p_{\rm L}(\mathcal{E})<p_{\rm G}(\mathcal{E})$.\\
\indent Corollary~\ref{cor:pql1} provides a useful method to make a bipartite quantum state ensemble $\mathcal{E}=\{\eta_{i},\rho_{i}\}_{i=1}^{n}$ showing nonlocality, that is, $p_{\rm L}(\mathcal{E})<p_{\rm G}(\mathcal{E})$, by means of DEW: For a given DEW $W$, let us consider
the bipartite quantum state ensemble $\mathcal{E}=\{\eta_{i},\rho_{i}\}_{i=1}^{2}$ where 
\begin{alignat}{3}\label{eq:dfex}
&\eta_{1}=\frac{\Tr(P+W)}{\Tr(2P+W)},~
\rho_{1}=\frac{P+W}{\Tr(P+W)},
\nonumber\\
&\eta_{2}=\frac{\Tr P}{\Tr(2P+W)},~
\rho_{2}=\frac{P}{\Tr P},
\end{alignat}
with any positive-semidefinite operator $P$ satisfying 
\begin{equation}
P+W\succeq0.
\end{equation}
\indent Since $\eta_{1}\rho_{1}-\eta_{2}\rho_{2}$
is proportional to the DEW $W$, $p_{\PPT}(\mathcal{E})<p_{\rm G}(\mathcal{E})$ holds from
Corollary~\ref{cor:pql1}.
Thus, Inequality~\eqref{eq:plptpg} leads us to
$p_{\rm L}(\mathcal{E})<p_{\rm G}(\mathcal{E})$.
In other words, Corollary~\ref{cor:pql1} provides a systematic way to construct 
a bipartite quantum state ensemble showing nonlocality in quantum state discrimination.\\
\indent Corollary~\ref{cor:pql1} can also be used to construct 
a bipartite quantum state ensemble $\mathcal{E}=\{\eta_{i},\rho_{i}\}_{i=1}^{n}$ with $n>2$ showing nonlocality in quantum state discrimination. For a set of DEWs $\{W_{i}\}_{i=2}^{n}$,
let us consider the bipartite quantum state ensemble $\mathcal{E}=\{\eta_{i},\rho_{i}\}_{i=1}^{n}$ where
\begin{alignat}{3}
&\eta_{1}=\frac{\Tr\mathbbm{1}}{\Tr(n\mathbbm{1}-\sum_{j=2}^{n}\lambda_{j}W_{j})},~
\rho_{1}=\frac{\mathbbm{1}}{\Tr\mathbbm{1}},
\nonumber\\
&\eta_{i}=\frac{\Tr(\mathbbm{1}-\lambda_{i}W_{i})}{\Tr(n\mathbbm{1}-\sum_{j=2}^{n}\lambda_{j}W_{j})},~
\rho_{i}=\frac{\mathbbm{1}-\lambda_{i}W_{i}}{\Tr(\mathbbm{1}-\lambda_{i}W_{i})},
~i=2,\ldots,n,\label{eq:dexd}
\end{alignat}
with any set of positive real numbers $\{\lambda_{i}\}_{i=2}^{n}$ satisfying 
\begin{equation}
\mathbbm{1}-\lambda_{i}W_{i}\succeq0~~\forall i=2,\ldots,n.
\end{equation}
Because $\eta_{1}\rho_{1}-\eta_{i}\rho_{i}$ is proportional to $W_{i}$ for any $i\in\{2,\ldots,n\}$,
$p_{\PPT}(\mathcal{E})<p_{\rm G}(\mathcal{E})$ holds from
Corollary~\ref{cor:pql1}.
Thus, Inequality~\eqref{eq:plptpg} leads us to
$p_{\rm L}(\mathcal{E})<p_{\rm G}(\mathcal{E})$.\\
\indent As our results establish a specific relation between the properties of EW and PPT measurements,
it is natural to investigate the relationship between EW and other measurements such as the set of all separable measurements.
It is also an interesting future work to investigate if EW can be used for the optimization of other state discrimination strategies besides minimum-error discrimination.\\ 

This work was supported by Basic Science Research Program(Grant No. NRF-2020R1F1A1A010501270) and Quantum Computing Technology Development Program(Grant No. NRF-2020M3E4A1080088) through a National Research Foundation of Korea grant funded by the Korea government(Ministry of Science and 
Information \& Communications Technology(ICT)).

\appendix
\section{Proof of Theorem~\ref{thm:pptq}}\label{app:thm1}
As we already have Inequality~\eqref{eq:pptq}, it is enough to show that 
\begin{equation}\label{eq:pqgt}
p_{\PPT}(\mathcal{E})\geqslant q_{\PPT}(\mathcal{E}).
\end{equation}

\begin{lemma}\label{lem:ppti}
If $E\in\mathbb{PPT}_{+}^{*}$ and $E\neq 0_{\mathbb{H}}$, then $\Tr E>0$, where $0_{\mathbb{H}}$ is the zero operator in $\mathbb{H}$.
\end{lemma}
\begin{proof}
The proof is by contradiction.
We first note that $\Tr E=\Tr(\mathbbm{1}E)\geqslant0$ because $\mathbbm{1}\in\mathbb{PPT}_{+}$ and $E\in\mathbb{PPT}_{+}^{*}$.
Thus, let us suppose $\Tr E=0$.\\
\indent For an arbitrary orthonormal product basis $\{|e_{i}\rangle\}_{i=1}^{D}$ of the bipartite Hilbert space $\mathcal{H}=\mathbb{C}^{d_{1}}\otimes\mathbb{C}^{d_{2}}$, we have
\begin{equation}\label{eq:tree2}
\sum_{i=1}^{D}\Tr(E|e_{i}\rangle\!\langle e_{i}|)=\Tr(E\mathbbm{1})
=\Tr E=0,
\end{equation}
where $D=d_{1}d_{2}$.
From $E\in\mathbb{PPT}_{+}^{*}$ and $|e_{i}\rangle\!\langle e_{i}|\in\mathbb{PPT}_{+}$ for all $i=1,\ldots,D$, we have
\begin{equation}\label{eq:tree1}
\Tr(E|e_{i}\rangle\!\langle e_{i}|)\geqslant0~~\forall i=1,\ldots,D.
\end{equation}
Equation~\eqref{eq:tree2} and Inequality~\eqref{eq:tree1} lead us to
\begin{equation}
\Tr(E|e_{i}\rangle\!\langle e_{i}|)=0
~~\forall i=1,\ldots,D.
\end{equation}
Since the choice of $\{|e_{i}\rangle\}_{i=1}^{D}$ can be arbitrary, 
\begin{equation}
\Tr(E|e\rangle\!\langle e|)=0
\end{equation}
for any product vector $|e\rangle\in\mathcal{H}$, therefore
\begin{equation}\label{eq:treff}
\Tr(EF)=0~~\forall F\in\mathbb{SEP},
\end{equation}
where $\mathbb{SEP}$ is defined in Eq.~\eqref{eq:sepdef}.\\
\indent We note that $\mathbb{SEP}$ spans $\mathbb{H}$. To see this, we first note that the set of all positive-semidefinite operators on $\mathbb{C}^{d_{k}}$ spans the set of all Hermitian operators on $\mathbb{C}^{d_{k}}$ for each $k=1,2$.
Moreover, every $F\in\mathbb{H}$ can be represented as a summation of product Hermitian operators, 
\begin{equation}
F=\sum_{l} A_{l}\otimes B_{l},
\end{equation}
where $A_{l}$ and $B_{l}$ are Hermitian operators on $\mathbb{C}^{d_{1}}$ and $\mathbb{C}^{d_{2}}$, respectively.
Therefore, Eq.~\eqref{eq:treff} leads us to 
\begin{equation}\label{eq:trefz}
\Tr(EF)=0~~\forall F\in\mathbb{H}.
\end{equation}
Equation~\eqref{eq:trefz} means $E=0_{\mathbb{H}}$ which contradicts the assumption of the lemma. Thus, $\Tr E>0$.
\end{proof}

\indent Let us consider the set
\begin{alignat}{3}
\mathcal{S}(\mathcal{E})=
\Big\{\big(\sum_{i=1}^{n}\eta_{i}\Tr(\rho_{i}M_{i})-p,\mathbbm{1}-\sum_{i=1}^{n}M_{i}\big)\in\mathbb{R}\times\mathbb{H}\,\Big|\, p>p_{\rm PPT}(\mathcal{E}),~
M_{i}\in\mathbb{PPT}_{+}~\forall i=1,\ldots,n\Big\},
\end{alignat}
where $\mathbb{R}$ is the set of all real numbers.
We note that $\mathcal{S}(\mathcal{E})$ is 
a convex set due to the convexity of $\mathbb{PPT}_{+}$ in Eq.~\eqref{def:pptpd}.
Moreover, $\mathcal{S}(\mathcal{E})$ does not have the origin $(0,0_{\mathbb{H}})$ of $\mathbb{R}\times\mathbb{H}$, 
otherwise there is a PPT measurement $\{M_{i}\}_{i=1}^{n}$ with 
\begin{equation}
\sum_{i=1}^{n}\eta_{i}\Tr(\rho_{i}M_{i})>p_{\PPT}(\mathcal{E}), 
\end{equation}
and this contradicts the optimality of $p_{\PPT}(\mathcal{E})$ in Eq.~\eqref{eq:pptdef}.
We also note that the Cartesian product $\mathbb{R}\times\mathbb{H}$
can be considered as a real vector space with an inner product defined as
\begin{equation}\label{eq:inpde}
\langle (a,A),(b,B)\rangle=ab+\Tr(AB)
\end{equation}
for $(a,A),(b,B)\in\mathbb{R}\times\mathbb{H}$.\\
\indent Since $\mathcal{S}(\mathcal{E})$ and the single-element set $\{(0,0_{\mathbb{H}})\}$ are disjoint convex sets,
it follows from the separating hyperplane theorem\cite{boyd2004,sht} that 
there is $(\gamma,\Gamma)\in\mathbb{R}\times\mathbb{H}$ satisfying
\begin{alignat}{1}
(\gamma,\Gamma)\neq(0,0_{\mathbb{H}}),\label{eq:gcgneq}\\
\langle(\gamma,\Gamma),(r,G)\rangle\leqslant0~~\forall(r,G)\in\mathcal{S}(\mathcal{E}).
\label{eq:rcgleq}
\end{alignat}
\indent Suppose 
\begin{alignat}{1}
\Tr\Gamma\leqslant\gamma p_{\PPT}(\mathcal{E}),\label{eq:sca1}\\
\Gamma-\gamma\eta_{i}\rho_{i}\in\mathbb{PPT}_{+}^{*}~\forall i=1,\ldots,n,\label{eq:sca2}\\
\gamma>0.\label{eq:sca3}
\end{alignat}
From Conditions~\eqref{eq:sca2} and \eqref{eq:sca3}, the Hermitian operator
$H=\Gamma/\gamma$ is an element of $\mathbb{H}_{\PPT}(\mathcal{E})$ in Eq.~\eqref{eq:hpts}.
Thus, the definition of $q_{\PPT}(\mathcal{E})$ in Eq.~\eqref{eq:qptdef} leads us to
\begin{equation}\label{eq:qptlet}
q_{\PPT}(\mathcal{E})\leqslant\Tr H.
\end{equation}
Moreover, Condition~\eqref{eq:sca1} implies
\begin{equation}\label{eq:thleqp}
\Tr H\leqslant p_{\PPT}(\mathcal{E}).
\end{equation}
Inequalities~\eqref{eq:qptlet} and \eqref{eq:thleqp} complete the proof of Inequality~\eqref{eq:pqgt}.\\
\indent The rest of this section is to prove Conditions~\eqref{eq:sca1}, \eqref{eq:sca2} and \eqref{eq:sca3}.
\begin{proof}[Proof of \eqref{eq:sca1}]
From Eq.~\eqref{eq:inpde},
Inequality~\eqref{eq:rcgleq} can be rewritten as
\begin{equation}\label{eq:trgmle}
\Tr \Gamma-\sum_{i=1}^{n}\Tr[M_{i}(\Gamma-\gamma\eta_{i}\rho_{i})]\leqslant \gamma p 
\end{equation}
for all $p>p_{\PPT}(\mathcal{E})$ and all $\{M_{i}\}_{i=1}^{n}\subseteq\mathbb{PPT}_{+}$. 
If $M_{i}=0_{\mathbb{H}}$ for all $i=1,\ldots,n$,  
Inequality~\eqref{eq:trgmle} becomes Inequality~\eqref{eq:sca1}
by taking the limit of $p$ to $p_{\PPT}(\mathcal{E})$.
\end{proof}
\begin{proof}[Proof of \eqref{eq:sca2}]
For each $j\in\{1,\ldots,n\}$,
let us consider an arbitrary $M_{j}\in\mathbb{PPT}_{+}$ 
and $M_{i}=0_{\mathbb{H}}$ for all $i=1,\ldots,n$ with $i\neq j$. 
In this case,
$\{M_{i}\}_{i=1}^{n}$ is clearly a subset of $\mathbb{PPT}_{+}$,
and Inequality~\eqref{eq:trgmle} becomes
\begin{equation}\label{eq:trmjle}
\Tr\Gamma-\Tr[M_{j}(\Gamma-\gamma\eta_{j}\rho_{j})]\leqslant\gamma p_{\PPT}(\mathcal{E})
\end{equation}
by taking the limit of $p$ to $p_{\PPT}(\mathcal{E})$.\\
\indent Suppose $\Gamma-\gamma\eta_{j}\rho_{j}\notin\mathbb{PPT}_{+}^{*}$,
then there is $M\in\mathbb{PPT}_{+}$ with $\Tr[M(\Gamma-\gamma\eta_{j}\rho_{j})]<0$. 
We note that $M\in\mathbb{PPT}_{+}$ implies $tM\in\mathbb{PPT}_{+}$ for any $t>0$. 
Thus, $\{M_{i}\}_{i=1}^{n}$ consisting of 
$M_{j}=tM$ for $t>0$ and $M_{i}=0$ for all $i=1,\ldots,n$ with $i\neq j$ 
is also a subset of $\mathbb{PPT}_{+}$.\\
\indent Now, Inequality~\eqref{eq:trmjle}
can be rewritten as
\begin{equation}\label{eq:trtmle}
\Tr\Gamma-\Tr[tM(\Gamma-\gamma\eta_{j}\rho_{j})]\leqslant\gamma p_{\PPT}(\mathcal{E}).
\end{equation}
Since Inequality~\eqref{eq:trtmle} is true for arbitrary large $t>0$, 
$\gamma p_{\PPT}(\mathcal{E})$ can also be arbitrarily large.
However, this contradicts that both  $\gamma$ and $p_{\PPT}(\mathcal{E})$ are finite.
Thus, $\Gamma-\gamma\eta_{j}\rho_{j}\in\mathbb{PPT}_{+}^{*}$, 
which completes the proof of \eqref{eq:sca2}.
\end{proof}
\begin{proof}[Proof of \eqref{eq:sca3}]
To show $\gamma\geqslant0$, we assume $\gamma<0$.
Let us consider $\{M_{i}\}_{i=1}^{n}$ with
$M_{i}=0_{\mathbb{H}}$ for all $i=1,\ldots,n$.
Since $\{M_{i}\}_{i=1}^{n}\subseteq\mathbb{PPT}_{+}$,
Inequality~\eqref{eq:trgmle} becomes
\begin{equation}
\Tr\Gamma\leqslant-\infty
\end{equation}
by taking the limit of $p$ to $\infty$.
This contradicts that $\Gamma$ is bounded. Thus $\gamma\geqslant0$.\\
\indent Now, let us suppose $\gamma=0$.
In this case, Conditions \eqref{eq:sca1} and \eqref{eq:sca2} become
\begin{equation}\label{eq:tglpp}
\Tr\Gamma\leqslant0,~~\Gamma\in\mathbb{PPT}_{+}^{*}.
\end{equation}
From Lemma~\ref{lem:ppti} together with Condition~\eqref{eq:tglpp}, we have 
\begin{equation}
\Gamma=0_{\mathbb{H}},
\end{equation}
which contradicts Condition~\eqref{eq:gcgneq}. Thus, $\gamma>0$.
\end{proof}



\begin{thebibliography}{99}
%
\bibitem{horo2009}
R.~Horodecki, P.~Horodecki, M.~Horodecki, and K.~Horodecki,
Quantum entanglement, 
\textit{Rev. Mod. Phys.} \textbf{81}, 865 (2009).
%
\bibitem{chid2013} 
A.~M.~Childs, D.~Leung, L.~Mančinska, and M.~Ozols,
A framework for bounding nonlocality of state discrimination,
\textit{Commun. Math. Phys.} \textbf{323}, 1121 (2013).
%
\bibitem{brun2014}  
N.~Brunner, D.~Cavalcanti, S.~Pironio, V.~Scarani, and S.~Wehner, 
Bell nonlocality, 
\textit{Rev. Mod. Phys.} \textbf{86}, 419 (2014).
%
\bibitem{chit20142}  
E.~Chitambar, D.~Leung, L.~Mančinska, M.~Ozols, and A.~Winter, 
Everything you always wanted to know about LOCC (but were afraid to ask), 
\textit{Commun. Math. Phys.} \textbf{328}, 303 (2014).
%
\bibitem{chef2000}  
A.~Chefles, 
Quantum state discrimination, 
\textit{Contemp. Phys.} \textbf{41}, 401 (2000).
%
\bibitem{berg2007}  
J.~A.~Bergou,
Quantum state discrimination and selected applications,
\textit{J. Phys.: Conf. Ser.} \textbf{84}, 012001 (2007).
%
\bibitem{barn20091}  
S.~M.~Barnett and S.~Croke, 
Quantum state discrimination,
\textit{Adv. Opt. Photonics} \textbf{1}, 238 (2009).
%
\bibitem{bae2015} 
J.~Bae and L.-C. Kwek, 
Quantum state discrimination and its applications, 
\textit{J. Phys. A: Math. Theor.} \textbf{48}, 083001 (2015).
%
\bibitem{benn19991}    
C.~H.~Bennett, D.~P.~DiVincenzo, C.~A.~Fuchs, T.~Mor, E.~Rains, P.~W.~Shor, J.~A.~Smolin, and W.~K.~Wootters, 
Quantum nonlocality without entanglement, 
\textit{Phys. Rev. A} \textbf{59}, 1070 (1999).
%
\bibitem{ghos2001} 
S.~Ghosh, G.~Kar, A.~Roy, A.~Sen(De), and U.~Sen,  
Distinguishability of Bell states, 
\textit{Phys. Rev. Lett.} \textbf{87}, 277902 (2001).
%
\bibitem{pere1991}  
A.~Peres and W.~K.~Wootters, 
Optimal detection of quantum information, 
\textit{Phys. Rev. Lett.} \textbf{66}, 1119 (1991).
%
\bibitem{chit2013}
E.~Chitambar and M.-H.~Hsieh, 
Revisiting the optimal detection of quantum information, 
\textit{Phys. Rev. A} \textbf{88}, 020302(R) (2013).
%
\bibitem{walg2000}
J.~Walgate, A.~J.~Short, L.~Hardy, and V.~Vedral, 
Local distinguishability of multipartite orthogonal quantum states,
\textit{Phys. Rev. Lett.} \textbf{85}, 4972 (2000).
%
\bibitem{virm2001}
S.~Virmani, M.~F.~Sacchi, M.~B.~Plenio, and D.~Markham,
Optimal local discrimination of two multipartite pure states,
\textit{Phys. Lett. A} \textbf{288}, 62 (2001).
%
\bibitem{lu2010}
Y.~Lu, N.~Coish, R.~Kaltenbaek, D.~R.~Hamel, S.~Croke, and K.~J.~Resch,
Minimum-error discrimination of entangled quantum states,
\textit{Phys. Rev. A} \textbf{82}, 042340 (2010).
%
\bibitem{akib2018}
S.~Akibue and G.~Kato,
Bipartite discrimination of independently prepared quantum states as a counterexample to a parallel repetition conjecture,
\textit{Phys. Rev. A} \textbf{97}, 042309 (2018).
%
\bibitem{ha20222}
D.~Ha and J.~S.~Kim,
Bound on optimal local discrimination of multipartite quantum states,
\textit{Sci. Rep.} \textbf{12}, 14130 (2022).
%
\bibitem{cohe2023}
S.~M.~Cohen,
Local approximation for perfect discrimination of quantum states,
\textit{Phys. Rev. A} \textbf{107}, 012401 (2023).
%
\bibitem{eker1991}
A.~K.~Ekert,
Quantum cryptography based on Bell's theorem, 
\textit{Phys. Rev. Lett.} \textbf{67}, 661 (1991).
%
\bibitem{benn1993}
C.~H.~Bennett, G.~Brassard, C.~Cr\'epeau, R.~Jozsa, A.~Peres, and W.~K.~Wootters,
Teleporting an unknown quantum state via dual classical and Einstein-Podolsky-Rosen channels, 
\textit{Phys. Rev. Lett.} \textbf{70}, 1895 (1993).
%
\bibitem{chit2019}
E.~Chitambar and G.~Gour,
Quantum resource theories,
\textit{Rev. Mod. Phys.} \textbf{91}, 025001 (2019).
%
\bibitem{band2021}  
S.~Bandyopadhyay and V.~Russo, 
Entanglement cost of discriminating noisy Bell states by local operations and classical communication,
\textit{Phys. Rev. A} \textbf{104}, 032429 (2021).
%
\bibitem{amic2008}  
L.~Amico, R.~Fazio, A.~Osterloh, and V.~Vedral,
Entanglement in many-body systems, 
\textit{Rev. Mod. Phys.} \textbf{80}, 517 (2008).
%
\bibitem{guhn2009}  
O.~Gühne and G.~Tóth, 
Entanglement detection, 
\textit{Phys. Rep.} \textbf{474}, 1 (2009).
%
\bibitem{kett2020}  
A.~Ketterer, N.~Wyderka, and O.~Gühne, 
Entanglement characterization using quantum designs, 
\textit{Quantum} \textbf{4}, 325 (2020).
%
\bibitem{horo1996}  
M.~Horodecki, P.~Horodecki, and R.~Horodecki, 
Separability of mixed states: necessary and sufficient conditions, 
\textit{Phys. Lett. A} \textbf{223}, 1 (1996).
%
\bibitem{terh2000}
B.~M.~Terhal,
Bell inequalities and the separability criterion,
\textit{Phys. Lett. A} \textbf{271}, 319 (2000).
%
\bibitem{lewe2000}  
M.~Lewenstein, B.~Kraus, J.~I.~Cirac, and P.~Horodecki, 
Optimization of entanglement witnesses, 
\textit{Phys. Rev. A} \textbf{62}, 052310 (2000).
%
\bibitem{chru2014}  
D.~Chruściński and G.~Sarbicki,
Entanglement witnesses: construction, analysis and classification, 
\textit{J. Phys. A: Math. Theor.} \textbf{47}, 483001 (2014).
%
\bibitem{pere1996}
A.~Peres, 
Separability criterion for density matrices, 
\textit{Phys. Rev. Lett.} \textbf{77}, 1413 (1996).
%
\bibitem{pptp}  
PPT property does not depend on the choice of basis or the subsystem to be transposed. 
For simplicity, we consider the standard basis and the second subsystem throughout this paper.
%
\bibitem{hels1969}
C.~W.~Helstrom, 
Quantum detection and estimation theory,
\textit{J. Stat. Phys.} \textbf{1}, 231 (1969).
%
\bibitem{hole1974}
A.~S.~Holevo,
Remarks on optimal quantum measurements,
\textit{Probl. Peredachi Inf.} \textbf{10}, 51 (1974).
%
\bibitem{yuen1975}  
H.~Yuen, R.~Kennedy, and M.~Lax, 
Optimum testing of multiple hypotheses in quantum detection theory, 
\textit{IEEE Trans. Inf. Theory} \textbf{21}, 125 (1975).
%
\bibitem{barn20092}  
S.~M.~Barnett and S.~Croke, 
On the conditions for discrimination between quantum states with minimum error, 
\textit{J. Phys. A: Math. and Theor.} \textbf{42}, 062001 (2009).
%
\bibitem{bae2013}  
J.~Bae, 
Structure of minimum-error quantum state discrimination, 
\textit{New J. Phys.} \textbf{15}, 073037 (2013).
%
\bibitem{cose2013}  
A.~Cosentino, 
Positive-partial-transpose-indistinguishable states via semidefinite programming, 
\textit{Phys. Rev. A} \textbf{87}, 012321 (2013).
%
\bibitem{walg2002} 
J.~Walgate and L.~Hardy, 
Nonlocality, asymmetry, and distinguishing bipartite states, 
\textit{Phys. Rev. Lett.} \textbf{89}, 147901 (2002).
%
\bibitem{deut1992} 
D.~Deutsch and R.~Jozsa, 
Rapid solution of problems by quantum computation, \textit{Proc. R. Soc. Lond. A} 
\textbf{439}, 553 (1992).
%
\bibitem{shor1994} 
P. Shor, Algorithms for quantum computation: discrete logarithms and factoring, 
in \textit{Proceedings 35th Annual Symposium on Foundations of Computer Science} (IEEE, 1994), pp. 124-134.
%
\bibitem{boyd2004} 
S.~Boyd and L.~Vandenberghe,
\textit{Convex Optimization} 
(Cambridge University Press, Cambridge, 2004).
%
\bibitem{sht}  
When $A$ and $B$ are disjoint convex sets in a real vector space $V$ with an inner product $\langle\cdot,\cdot\rangle$,
there exist $x\in\mathbb{R}$ and $\vec{v}\in V$ such that $\vec{v}\neq\vec{0}$ and $\langle \vec{a},\vec{v}\rangle\leqslant x\leqslant\langle \vec{b},\vec{v}\rangle$  for all $\vec{a}\in A$ and all $\vec{b}\in B$.
\end{thebibliography}
\end{document}